\documentclass[11pt]{article}

\usepackage{mathptmx}
\usepackage{amsmath,amssymb,amsfonts,amsthm} 
\usepackage{mathrsfs,latexsym} 
\usepackage{mathtools}

\usepackage{calrsfs}
\DeclareMathAlphabet{\pazocal}{OMS}{zplm}{m}{n}

\usepackage{color}
\usepackage{cite}
\usepackage{epsfig}
\usepackage{graphicx}

\usepackage{bm}

\usepackage{cite}

\newtheorem{theorem}{Theorem}[section]

\newtheorem{lemma}[theorem]{Lemma}
\newtheorem{proposition}[theorem]{Proposition}

\numberwithin{equation}{section}

\newcommand{\RR}{{\mathbb R}}

\newcommand{\ZZ}{{\mathbb Z}}

\newcommand{\Cc}{{\mathcal{C}}}
\newcommand{\Dc}{{\mathcal{D}}}

\newcommand{\Hc}{{\mathcal{H}}}

\newcommand{\Oc}{{\mathcal{O}}}

\newcommand{\balpha}{{\bm \alpha}}

\newcommand{\bgamma}{{\bm \gamma}}

\newcommand{\bix}{{\bm x}}

\newcommand{\fP}{{\mathfrak P}}

\newcommand{\Times}{\mbox{\Large $\times$}}      

\newcommand{\supp}{{\mathrm{supp}}}

\def\ie{{\it i.e.\ }}

\begin{document} 

%

\title{\LARGE Linking numbers in local quantum field theory}
\author{Detlev Buchholz${}^a$,
Fabio Ciolli${}^b$, Giuseppe Ruzzi${}^b$ and   
Ezio Vasselli${}^b$ \\[20pt]
\small  
${}^a$ Mathematisches Institut, Universit\"at G\"ottingen, \\
\small Bunsenstr.\ 3-5, 37073 G\"ottingen, Germany\\[5pt]
\small
${}^b$ 
Dipartimento di Matematica, Universit\'a di Roma ``Tor Vergata'' \\
\small Via della Ricerca Scientifica 1, 00133 Roma, Italy \\
}
\date{}

\maketitle

{\small 
  \noindent {\bf Abstract.}
    Linking numbers appear in local quantum field theory in the 
    presence of tensor fields, which are closed two-forms on Minkowski space.
    Given any pair of such fields, it is shown that the commutator of 
    the corresponding intrinsic (gauge invariant) vector potentials, 
    integrated about spacelike separated, spatial loops, 
    are elements of the center of the algebra of all local fields. 
    Moreover, these 
    commutators are proportional to the linking numbers of the underlying 
    loops. If the commutators 
    are different from zero, the underlying two-forms are not exact 
    (there do not exist local vector potentials for them). 
    The theory then necessarily contains massless particles.
    A prominent example of this kind, due to J.E.~Roberts, is given 
    by the free electromagnetic field and its Hodge dual. Further 
    examples with more complex mass spectrum
    are presented in this article. 
  } \\[2mm]
 \\[2mm]

\section{Introduction}

Skew symmetric tensor fields, which are closed 
two-forms on Minkowski space, are 
familiar from local quantum field theory, the most 
prominent example being the electromagnetic field satisfying the 
homogeneous Maxwell equation. One frequently argues that such  
two-forms are exact and introduces 
corresponding local vector potentials, which are in general 
defined on indefinite metric spaces. It was pointed out 
by J.E.~Roberts \cite{Roberts} that this step might not always be   
possible, however, if one is dealing with several  
closed two-forms, such as the free electromagnetic field and its 
Hodge dual.  

\medskip
Obstructions to the existence of local vector potentials manifest 
themselves in commutators of the given tensor fields, which are integrated 
over \mbox{compact} two-surfaces with spacelike separated boundaries. 
These integrals can be rewritten in terms of intrinsically defined 
(gauge invariant) vector potentials, 
cf.\ \cite{BuCiRuVa1}, which are integrated along the boundaries.
If these commutators are different from zero, the intrinsic 
vector potentials, being defined on loop functions, cannot be 
extended to vector potentials which are point-like and local 
with respect to each other. Phrased differently, the underlying closed 
two-forms can not be exact in the setting of local quantum field
theory. Let us explain this point in somewhat more detail. 

\medskip 
Let $F_{\mu \nu}(f^{\mu \nu})$ be a local, hermitean, skew symmetric 
rank-two tensor field in the framework
of quantum field theory. It is assumed to be  
linear with regard to the tensor-valued test functions $f^{\mu \nu}$ and 
closed as a two-form on Minkowski space, \ie 
$F_{\mu \nu}(\partial_\rho f^{\mu \nu \rho}) = 0$, 
where $f^{\mu \nu \rho}$ are test functions with values in skew symmetric
rank-three tensors and $\partial_\rho$ are the
spacetime derivatives. There then exists
a corresponding intrinsic (gauge invariant) vector potential $A_\mu(h^\mu)$.
It is defined for all co-closed vector-valued 
test functions $h^\mu$, satisfying 
$\partial_\mu h^\mu = 0$, which form a 
vector space denoted by $\Cc_1(\RR^4)$. 
Given $h \in \Cc_1(\RR^4)$, the relation
between the tensor field and the intrinsic vector potential 
is established by the formula
$F_{\mu \nu}(f^{\mu \nu}) = A_\mu(h^\mu)$, where 
$f^{\mu \nu}$ is any test function satisfying
\mbox{$\partial_\nu f^{\mu \nu} = h^\mu$}.
Such functions exist by the Poincar\'e lemma and $F_{\mu \nu}(f^{\mu \nu})$
depends only on $h^\mu$ in view of the fact that 
$F_{\mu \nu}$ is a closed two-form, cf.\ \cite{BuCiRuVa1}.

\medskip
Einstein causality is expressed by the condition of locality 
according to which the commutator of local fields vanishes
at spacelike distances. The intrinsic vector potentials inherit
from the underlying local tensor fields certain locality properties:
if $h \in \Cc_1(\RR^4)$ is any test function which has support in some
bounded, contractible region $\Oc$, then $A_\mu(h^\mu)$ 
commutes with all local operators which are localized in similar 
regions in the spacelike complement of $\Oc$. Yet if the 
support properties of $h$ are topologically non-trivial,
this commutativity may fail. Consider for example  
two local, hermitean, skew symmetric and closed 
rank-two tensor fields $F_{\mu \nu}$, $G_{\rho \sigma}$
with corresponding intrinsic vector potentials $A_\mu$ and
$B_\rho$, respectively. As was shown in 
\cite{Roberts,BuCiRuVa2} 
by examples, the commutator 
\begin{equation}
\label{eq.comm}
[A_\mu(h^\mu), B_\rho(k^\rho)]
\end{equation}
need not vanish for test functions $h, k \in \Cc_1(\RR^4)$
having support in linked, spacelike separated loop-shaped 
regions. The intrinsic vector potentials can then
not be extended to local point fields 
which are defined on the space of 
all vector-valued test functions $\Dc_1(\RR^4) \supset \Cc_1(\RR^4)$. 

\medskip
In the present article we study the properties 
of the commutators $[A_\mu(h^\mu), B_\rho(k^\rho)]$ 
of intrinsic vector potentials in detail and refer to them as
\textit{causal commutators} if the supports of the
underlying test functions are spacelike \mbox{separated}. 
Our analysis is organized as follows. 
Denoting by $\fP$ the \mbox{polynomial} *-algebra generated by 
all local fields in the theory and making use of results in 
\cite{BuCiRuVa1}, we will show in a first step that the commutators are
central (superselected) elements of $\fP$ whenever the supports of 
$g$ and $h$ are spacelike separated. For certain specific 
test functions, having supports in linked, spacelike separated 
loop-shaped regions (loop functions), the causal commutators are 
then shown to be stable under deformations of the loops.
Thus these commutators encode topological information
and are proportional to linking numbers associated
with the underlying loops. Similar results were stated in 
\cite{Roberts} without proofs; these are supplied here. 
Our main result consists 
of the demonstration that the causal commutators can be different from zero 
only in theories involving massless particles. 
This resembles Swieca's theorem on the existence 
of massless particles in theories with an electric charge
\cite{Sw}, but it is not directly related to it. 
Examples of theories with non-trivial causal commutators 
are then exhibited and the article closes with some 
expository remarks. 

\section{Causal commutators are elements of the center}
\label{sec2}

In this section we analyze algebraic properties of the   
causal commutators 
\eqref{eq.comm}. To this end, we recall some standard
notation for skew symmetric 
tensor functions in order to express  
properties of their quantum counterparts in precise mathematical terms. 

\medskip
In the sequel, we shall adopt the short hand notation $t \in \Dc_n(\RR^4)$
for skew symmetric $n$-tensor valued test functions
$t \doteq (t^{\mu_1 \ldots \mu_n})$, $n = 0, \dots, 4$. Considering their 
supports, we write $\supp(t) \perp \supp(t')$ to denote 
the spacelike separation of the supports 
of pairs $t,t' \in \Dc_n(\RR^4)$ in Minkowski spacetime. 
As already mentioned in the introduction, a stronger condition 
of spacelike separation is obtained by requiring that the 
supports of $t,t'$  are contained in 
spacelike separated, bounded, contractible regions,
which then are separated by two characteristic planes; 
in this case we write
$\supp(t) \, \Times \, \supp(t')$.
The difference between $\perp$ and $\Times$ becomes manifest when 
the supports of $t,t'$ have non-trivial topological properties.

\medskip
On the tensor valued test functions acts 
the co-derivative $\delta$, whose action on  
\mbox{$f \in \Dc_2(\RR^4)$} is given by  
$(\delta f)^\mu \doteq -2\partial_\nu f^{\nu \mu} \in \Cc_1(\RR^4)$,
where $\Cc_1(\RR^4)$ consists of all co-closed
functions in \mbox{$\Dc_1(\RR^4)$}, satisfying   
$\delta h \doteq - \partial_\nu h^\nu = 0$.
According to the Poincar\'e Lemma one has 
$\Cc_1(\RR^4) = \delta(\Dc_2(\RR^4))$, 
cf.\ \cite{BuCiRuVa1}.

\medskip 
Passing to the quantum level, we 
consider hermitean operators 
$F(f)$, $G(g)$,
linear in $f,g \in \Dc_2(\RR^4)$, 
which fulfill the homogeneous Maxwell equation
\begin{equation}
\label{eq.maxw}
dF(t) \, \doteq \, F_{\mu\nu}(\partial_\rho t^{\rho\mu\nu}) \, = \, 0 \, ,
\quad dG(t) \, = \, 0 \, , \quad t \in \Dc_3(\RR^4) \, ,
\end{equation}
and the causality relations
\begin{equation}
\label{eq.caus}
[ F(f) \, , \, G(g) ] \, = \, 0 \, , \quad 
\supp(f) \perp \supp(g) \, , \quad f,g \in \Dc_2(\RR^4) \, .
\end{equation}
The homogeneous Maxwell equation implies that for any
$h, k \in \Cc_1(\RR^4)$ the operators
$$
A(h) \, \doteq   \, F(f) \, , \quad B(k) \, \doteq \, G(g)
$$
are well-defined for any choice of 
test functions $f,g \in \Dc_2(\RR^4)$, satisfying  
$\delta f = h$ and $\delta g = k$, c.f.\ \cite{BuCiRuVa1}. 
We call $A$ and $B$ the \emph{intrinsic vector potentials} 
defined by $F$ and $G$, respectively.

\medskip
We assume that the tensor fields $F, G$ are elements of some 
polynomial *-algebra $\fP$, which is generated by local field operators
and on which the Poincar\'e group acts covariantly by automorphisms. 
Based on this input, we want to clarify the properties of commutators of 
the corresponding intrinsic vector potentials, cf \eqref{eq.comm}. 
We begin with a technical lemma which slightly generalizes a result 
in \cite{BuCiRuVa1}. 
\begin{lemma}
\label{lem.2.1}
Let $h \in \Cc_1(\RR^4)$ and $g \in \Dc_2(\RR^4)$ such that  
$\supp(g) \perp \supp(h)$. Then
$$
[ A(h) , G(g) ]  =  [B(h),F(g)]=0 \   .
$$
\end{lemma}

\begin{proof}
Using a partition of unity we can decompose $g$ into a a finite sum 
$\sum_i g_i$ where $g_i$ is supported in a double cone $\Oc_i$ such 
that $\Oc_i \perp \supp (h)$ for any $i$.
By the Causal Poincar\'e Lemma (see \cite[Appendix]{BuCiRuVa1}) 
we can find for any $i$ a co-primitive $f_i$ of $h$,
$\delta f_i = h$, such that 
$\supp(f_i) \perp \Oc_i$. Hence, by \eqref{eq.caus}, one has  
$$
[ A(h) , G(g)] = \sum_i [A(h) , G(g_i)] = \sum_i [F(f_i) , G(g_i)]=0 \, .
$$ 
The same argument applies to the commutator between $B$ and $F$. 
\end{proof}

Using this result it is easily seen that the causal commutator of the 
intrinsic vector potentials, defined on test functions with
spacelike separated supports,   
is a central element, hence a superselected quantity.
The subsequent proposition is a slight generalization of 
results in \cite[Appendix]{BuCiRuVa1}.

\begin{proposition}
\label{prop.2.2}
Let $h,k \in \Cc_1(\RR^4)$  with $\supp(h) \perp \supp(k)$. 
Then the com\-mutator $[ A(h) , B(k) ]$ is translation invariant 
and hence an element of the center of~$\fP$. 
\end{proposition}
\begin{proof}
Let $h \in \Cc_1(\RR^4)$ and let $h_y$ be its 
translate for any $y \in \RR^4$. We proceed to the test function
$f^{\, y} \in \Dc_2(\RR^4)$ given by 
$$
f^{\, y \, \mu \nu}(x)   
\doteq (1/2)\,\int_0^1 \! du \, (x^{\mu} h^{\nu}(x - uy) 
- y^{\mu} h^{\nu}(x - u y)) \ , \quad x\in\RR^4 \, . 
$$
It is a co--primitive of $(h - h_y) \in \Cc_1(\RR^4)$, that is,
$\delta f^{\, y} = (h - h_y)$.
Moreover, $f^{\, y}$ has support in the cylindrical region 
$\{ \supp(h) + u y : 0 \leq u \leq 1 \}$. So, for sufficiently small
translations $y \in \RR^4$, we have $\supp(f^{\, y}) \perp \supp(k)$.
This implies by Lemma \ref{lem.2.1} 
$$
\begin{array}{lcl}
[A(h), B(k)] & = & [A(\delta f^{\, y} +h_y), B(k)] \\ & = &  
[F(f^{\, y}), B(k)] + [A(h_y), B(k)] \\ & = &
[A(h_y), B(k)] \, . 
\end{array}
$$
Applying the same argument to the
translates of $k$ one finds that 
$$
[A(h), B(k)] = [A(h_{y}), B(k_{y})]
$$
for sufficiently small $y$. This equality extends to arbitrary translations 
$y \in \RR^4$ by iteration. 
It then follows from translation covariance, locality, 
and the Jacobi identity that,
for any given $h, k \in \Cc_1(\RR^4)$ 
with $\supp(h) \perp \supp(k)$,
the commutator $[A(h), B(k)]$ commutes with any other element of $\fP$.
Hence, being itself an element of $\fP$,
it lies in the center of $\fP$, completing the proof.
\end{proof}

\section{Homology invariance of causal commutators}
\label{sec3}

Turning to the geometrical analysis of the causal commutators, 
we recall from \cite[Sect.~3]{BuCiRuVa2} some definitions 
and facts about loops, surfaces and test functions.
Given a test function $s \in \Dc_0(\RR^4)$ and a 
loop $\gamma : [0,1] \to \RR^4$, one can define 
a test function in $l_{s,\gamma} \in \Cc_1(\RR^4)$, 
called \emph{loop function}, putting 
$$
l_{s,\gamma}^\mu(x) \doteq \int_0^1 \! du \, s( x + \gamma(u) ) \, 
\dot{\gamma}^\mu(u) \, ,
$$
where $\dot{\gamma}^\mu$ is the tangent vector. It is apparent that 
$\supp(l_{s,\gamma}) \subset \supp(s) + \gamma$. One can perform 
a similar construction for surfaces 
$\sigma : [0,1] \times [0,1] \rightarrow \RR^4$ and define
corresponding test functions $f_{s,\sigma} \in \Dc_2(\RR^4)$, putting 
$$
f_{s,\sigma}^{\mu\nu}(x) \, \doteq \, 
- (1/2) \, \int_0^1 \! \! \int_0^1 \! d^2u \, s( x + \sigma(u) ) \, 
\sigma^{\mu\nu}(u) \, ,
$$
where $\sigma^{\mu\nu}$ is the Jacobian fixed by $\sigma$. 
Clearly, $\supp(f_{s,\sigma}) \subset \supp(s) + \sigma$.
One then has the relation 
\begin{equation}
\label{eq.coStokes}
\delta f_{s,\sigma} \ = \ l_{s,\partial \sigma} \, ,
\end{equation}
where $\partial \sigma$ is the boundary of $\sigma$. 
Since this is a straight forward consequence of Stokes' theorem,
its proof is omitted.

\medskip
In the next step we prove two invariance properties
of the commutators of interest here. 
The first one concerns the dependence of the commutator
on the co-cohomology class of the loop functions 
entering in the smearing of the operators; to be precise:  
the 0-co-cohomology class of the underlying scalar functions
$s \in \Dc_0(\RR^4)$, which is fixed by 
its integral $\int \! dx \, s(x)$. 

\begin{lemma}
\label{lem.3.1}
Let $\gamma_1,\gamma_2$ be two spacelike separated loops, 
let $\Oc_1,\Oc_2$  be two open balls centered about the origin 
such that $\Oc_1 + \gamma_1 \perp \Oc_2 + \gamma_2$, 
and let $s_1,s_2 \in \Dc_0(\RR^4)$ be test 
functions with support in $\Oc_1$, respectively $\Oc_2$.  
Then, for any test function $\widehat s_1$ 
with $\supp(\widehat s_1) \subset \Oc_1$ and  
$\int \! dx \, \widehat s_1(x) = 1$, one has 
$$
[ A(l_{s_1,\gamma_1})  , B(l_{s_2,\gamma_2}) ] 
\ = \ \kappa \,  [ A(l_{\widehat s_{1},\gamma_1})  , B(l_{s_2,\gamma_2}) ]  \, ,
$$
where 
$\kappa \doteq \int \! dx \, s_1(x)$.
\end{lemma}
\begin{proof}
Since the integral of the function 
$(s_1-\kappa\, \widehat s_1)$ is zero, there is a 
test function $h \in \Dc_1(\RR^4)$ 
such that $\supp(h) \subset \Oc_1$ and 
$\delta h = (s_1 - \kappa \, \widehat s_1)$. Thus
the associated test function $f \in \Dc_2(\RR^4)$,  
given by
$$
 f^{\mu\nu}(x) \doteq (1/2) \int_0^1 \! du \, 
\left(h^\mu(x-\gamma_1(u)) \, \dot{\gamma}^\nu_1(u)- 
h^{\nu}(x-\gamma_1(u)) \, \dot{\gamma}^\mu_1(u)\right)
\, , 
$$   
has support in 
$\Oc_1 + \gamma_1$ and  
$\delta f= (l_{s_1,\gamma_1} - \kappa \, l_{\widehat s_1,\gamma_1})$. 
According to Lemma \ref{lem.2.1}, this implies 
\begin{align*}
[ A(l_{s_1,\gamma_1})  , B(l_{s_2,\gamma_2})] & = 
\kappa\,[ A(l_{\widehat s_1,\gamma_1})  , B(l_{s_2,\gamma_2})] + 
[ A(\delta f)  , B(l_{s_2,\gamma_2})] \\
& =
\kappa\,[ A(l_{\widehat s_1,\gamma_1})  , B(l_{s_2,\gamma_2}) + 
[ F(f)  , B(l_{s_2,\gamma_2})] \\ 
& = \kappa\,[ A(l_{\widehat s_1,\gamma_1})  , B(l_{s_2,\gamma_2})] \, , 
\end{align*} 
completing the proof.
\end{proof}

\noindent Since the commutator in the lemma vanishes if either 
$\int \! dx  \, s_1(x) =0 $ or $\int \! dx \, s_2(x) =0$, 
we adopt below the following convention.

\medskip 
\noindent \textbf{Standing assumption:} 
The scalar functions $s \in \Dc_0(\RR^4)$ entering in the loop 
functions~$l_{s,\gamma}$ are normalized, \ie $\int \! dx \, s(x) = 1$. 

\medskip
The second property of the causal commutators 
we want to establish is their homology invariance
with regard to deformations of the underlying loops. 

\begin{lemma}
\label{lem.3.2}
Let $\gamma_1,\gamma_2$ be two spacelike separated loops, 
let $\Oc_1,\Oc_2$  be open balls centered about the origin 
such that $\Oc_1 + \gamma_1 \perp \Oc_2 + \gamma_2$,  
and let $s_1,s_2 \in \Dc_0(\RR^4)$ be normalized  
functions with support in $\Oc_1$, respectively $\Oc_2$. 
Then, for any loop $\widehat \gamma_1$ which is homologous to 
$\gamma_1$ in the causal complement of $\gamma_2$
and any normalized function~$\widehat s_1$ 
having support in an open ball $\widehat \Oc_1$ 
about the origin such that $\widehat \Oc_1 + \widehat \gamma_1 
\perp \Oc_2 + \gamma_2$, one has 
$$
[ A(l_{s_1,\gamma_1})  , B(l_{s_2,\gamma_2}) ] 
\, = \,
[ A(l_{\widehat s_1,\widehat \gamma_1}) , B(l_{s_2,\gamma_2}) ] \, .
$$
An analogous relation holds if 
$(s_2, \gamma_2)$ is replaced by a homologous  
pair~$(\widehat s_2, \widehat \gamma_2)$.  
\end{lemma}
\begin{proof}
According to the  hypothesis on $\widehat \gamma_1$,  
there is a surface $\sigma \subset \RR^4$ such that 
$\sigma \perp \gamma_2$ and 
$\partial \sigma = \gamma_1 - \widehat \gamma_1 $. 
Let $\widehat \Oc_1$ be any open ball
about the origin such that 
$\widehat \Oc_1 \subset \Oc_1$ and 
$\widehat \Oc_1 + \sigma \perp \Oc_2 + \gamma_2$, and pick  
a function $\widehat s_1$ with support in 
$\widehat \Oc_1$. Then, by Lemma \ref{lem.3.1}, we have 
$$
[ A(l_{s_1,\gamma_1})  , B(l_{s_2,\gamma_2}) ] 
\, = \,
[ A(l_{\widehat s_1,\gamma_1}) , B(l_{s_2,\gamma_2}) ] \, . 
$$
As in the preceding lemma, we consider the 
test function $f_{\widehat s_1,\sigma} \in \Dc_2(\RR^4)$. It has support 
in $\widehat \Oc_1 + \sigma$, and 
$\delta f_{\widehat s_1,\sigma} =   
l_{\widehat s_1, \widehat \gamma_1} - 
l_{\widehat s_1, \gamma_1}$ by  
$\eqref{eq.coStokes}$.
Since $\widehat \Oc_1 + \sigma\perp \Oc_2+\gamma_2$,  
it follows from Lemma \ref{lem.2.1} that   
\begin{align*}
[ A(l_{s_1,\gamma_1})  , B(l_{s_2,\gamma_2}) ] 
 & = 
[ A(l_{\widehat s_1,\gamma_1}) , B(l_{s_2,\gamma_2}) ] \\
& = 
[ F(f_{\widehat s_1,\sigma}) , B(l_{s_2,\gamma_2}) ] + 
[ A(l_{\widehat s_1, \widehat \gamma_1}) , B(l_{s_2,\gamma_2}) ]\\
 & = 
[ A(l_{\widehat s_1, \widehat \gamma_1}) , B(l_{s_2,\gamma_2}) ] \, ,
\end{align*}
proving the Lemma.
\end{proof}

\medskip
It follows from this lemma that the commutators are stable under
deformations of the underlying loop functions within the given 
limitations. With this information we can turn now to the 
analysis of their relation to the linking numbers, 
where we begin by recalling some topological definitions.
For any simple loop $\bgamma$ in $\RR^3$, the first homology group  
$H_1(\RR^3 \backslash \bgamma)$ is isomorphic to~$\ZZ$ 
(Alexander duality, \cite[Chpt.~3]{Hatcher}).
Given any pair of disjoint linked loops 
$\bgamma_1 \, , \bgamma_2 \subset \RR^3$, where $\bgamma_1$ is simple, 
their linking 
number  $L(\bgamma_1, \bgamma_2)$ is equal to  the homology 
class of $\bgamma_2$ in  $H_1(\RR^3 \backslash \bgamma_1) \simeq \ZZ$,
respectively of $\bgamma_1$ in  $H_1(\RR^3 \backslash \bgamma_2)$, 
cf.~\cite[Chpt.~1]{Hatcher}, \cite[Chpt.~5]{Rolfsen}. 
It coincides with the 
analytical definition given in terms of the \emph{Gauss integral},  
\begin{equation}
\label{eq.gauss}
L(\bgamma_1, \bgamma_2) \, \doteq \,
(1/4 \pi) \, 
\int_0^1 \! du \! \int_0^1 \! dv \, 
\frac
{ \det \, ( \dot{\bgamma_1}(u) \, , \, \dot{\bgamma_2}(v) 
\, , \, \bgamma_1(u) - \bgamma_2(v) ) }
{ | \bgamma_1(u) - \bgamma_2(v) |^3 } \, .
\end{equation}

\medskip
Although Lemma \ref{lem.3.2} implies that the commutators of the 
intrinsic vector potentials $A$, $B$ are stable under
deformations of the loop functions, there is some 
subtle point to be observed. It is known that
loops in  $\RR^4 \backslash \gamma$ can be disentangled 
(trivialized), namely $H_1(\RR^4 \backslash \gamma) = 0$. 
Yet the homology invariance of 
Lemma \ref{lem.3.2} refers to the \emph{causal
complement} and not to the complement of the curves in $\RR^4$. 
A direct connection between this notion and the 
notion of linking numbers is obtained 
by considering a particular class of loops in Minkowski space.
A loop $\gamma:[0,1]\to\RR^4$ is said to be \emph{spatial} whenever 
the points $\gamma(u)$, $\gamma(v)$ are spacelike with respect
to each other for any $u \ne v$ 
and $(u,v) \ne (0,1),(1,0)$. Such loops are simple and so are 
their projections onto the time zero plane, given by  
$\RR^4 \supset \gamma = (\gamma_0, \bgamma) \mapsto \bgamma \subset \RR^3$. 

\medskip 
Given any two spacelike separated,  
spatial loops $\gamma_1 \perp \gamma_2$, one can continuously deform 
both of them to disjoint loops in the 
time-zero plane~$(0,\RR^3)$, without affecting their 
spacelike separation. This is accomplished by the function 
$H: [0,1] \times \RR^4 \rightarrow \RR^4$, given by 
$$
H(u,(x_0,\bix)) \doteq ((1-u)x_0,\,\bix) 
\, , \quad  u \in [0,1] \, , \, 
(x_0,\,\bix) \in \RR^4 \, . 
$$
Clearly, $H$ is continuous, 
$H(0,\, \cdot \,) \upharpoonright \RR^4 
= id$ and 
$H(1,\, \cdot \,) \upharpoonright \RR^4
= P_{\, 0}$, the projection onto $(0,\RR^3) \subset \RR^4$. Putting 
$\gamma_{1,u} \doteq H(u,\gamma_1)$, $\gamma_{2,u} \doteq H(u,\gamma_2)$, 
one has $\gamma_{1,u} \perp \gamma_{2,u}$ for any $u \in [0,1]$ and
$\gamma_{1,1} = (0, \bgamma_1)$, $\gamma_{2,1} = (0, \bgamma_2)$
are disjoint loops lying in $(0, \RR^3)$.  
These observations allow us to introduce 
in $\RR^4$ a \textit{causal} linking number for spacelike separated,
spatial loops $\gamma_1$, $\gamma_2$, given by
\begin{equation}
\label{linking}
L_c(\gamma_1,\gamma_2) \, \doteq \, 
L(\bgamma_1, \bgamma_2) \, \in H_1(\RR^3 \backslash \bgamma_2) \simeq \ZZ \, .
\end{equation}

We can establish now the main result of this section, 
stating the proportionality of the causal commutators of intrinsic 
vector potentials to the causal linking numbers of the underlying loops. 
\begin{proposition}
\label{prop.3.3}
Let $\gamma_1,\gamma_2$ be spacelike separated, spatial loops,  
let $\Oc_1,\Oc_2$  be open balls centered about the origin 
such that $\Oc_1 + \gamma_1 \perp \Oc_2+\gamma_2$, and 
let $s_1,s_2 \in \Dc_0(\RR^4)$ be normalized  
functions with support in $\Oc_1$, respectively~$\Oc_2$. Then 
$$
[ A(l_{s_1,\gamma_1}) , B(l_{s_2,\gamma_2}) ] 
\ = \ 
i \, L_c(\gamma_1,\gamma_2) \ Z_{A,B} \, ,
$$
where $L_c(\gamma_1,\gamma_2) \in \ZZ$ is the causal linking
number of $\gamma_1, \gamma_2$ and 
$Z_{A,B} = - Z_{B,A}$ is a fixed hermitean element of the center of $\fP$
which does not depend on the choice of the loops and
test functions within the above limitations. 
\end{proposition}
\begin{proof}
According to Lemma \ref{lem.3.2} and the preceding 
remark concerning the deformation of spacelike separated,
spatial loops, one has 
$$ [ A(l_{s_1,\gamma_1}) , B(l_{s_2,\gamma_2}) ] 
\ = \ [ A(l_{\widehat s_1,\bgamma_1}) , B(l_{\widehat s_2,\bgamma_2}) ] \, ,
$$
where the normalized test functions $\widehat s_1, \widehat s_2$ 
have small enough supports, as required by Lemma \ref{lem.3.2}. Let 
$$
\lambda \doteq L_c(\gamma_1,\gamma_2) = 
L(\bgamma_1, \bgamma_2) \in \ZZ 
$$ 
be the (without loss of generality positive) 
causal linking number for the given 
loops $ \gamma _1,  \gamma _2$. The
projected loop $\bgamma_1$ is homologous in 
$\RR^3 \backslash \bgamma_2$ to the 
$\lambda$-fold composition $\balpha_\lambda  
\doteq \balpha_1*\cdots* \balpha_1$ 
of a generating circle 
$\balpha_1$ of the homology group $H_1(\RR^3 \backslash \bgamma_2)$,  
\ie $L(\balpha_1, \bgamma_2) = 1$. 
Thus, by another application of Lemma~\ref{lem.3.2}, one gets  
$$
 [ A(l_{\widehat s_1, \bgamma_1}) , B(l_{\widehat s_2,\bgamma_2}) ] 
= \lambda \, [ A(l_{{\widehat s}^{\, \prime}_1, \balpha_1}) , 
B(l_{{\widehat s}^{\, \prime}_2,\bgamma_2})] ,
$$
where ${\widehat s}^{\, \prime}_1$, ${\widehat s}^{\, \prime}_2$
are test functions
complying with the support conditions in the lemma.
The appearance of the factor $\lambda$ follows from the fact 
that loop functions are invariant under changes of the parametrization, 
so \mbox{$l_{\widehat s_1,\balpha_\lambda} = \lambda \, l_{\widehat s_1,\balpha_1}$}. 
Now $\bgamma_2$ can in turn be regarded as generator of 
the homology group $H_1(\RR^3 \backslash \balpha_1)$, hence it 
is homologous in  $\RR^3 \backslash \balpha_1$
to a circle $\balpha_2$. By a final application of 
Lemma~\ref{lem.3.2}, we therefore arrive at 
$$
[ A(l_{{\widehat s}^{\, \prime}_1, \balpha_1}) , 
B(l_{{\widehat s}^{\, \prime}_2, \bgamma_2})]
= [ A(l_{s , \balpha_1}) , 
B(l_{s ,\balpha_2})] 
$$
with a normalized test function $s$, 
satisfying the support conditions. 
Note that the circles $\balpha_1$, $\balpha_2$   
and the function~$s$ can be chosen independently of 
the initial data $\gamma_1$, $\gamma_2$ and $s_1$, $s_2$; in particular 
$[ A(l_{s, \balpha_1}) , B(l_{s, \balpha_2})] = [ A(l_{s, \balpha_2}) , B(l_{s, \balpha_1})]$.
The operator 
\mbox{$Z_{A,B} \doteq -i \, [ A(l_{s , \balpha_1}) , B(l_{s ,\balpha_2})] 
= - Z_{B,A}$} is contained in the center
of~$\fP$, cf.\ Lemma~\ref{lem.2.1}, and 
$Z_{A,B}^* = Z_{A,B}$ since $A,B$ are hermitean.
The statement then follows from the preceding equalities. 
\end{proof}
This proposition shows that the causal commutators of intrinsic
vector potentials, smeared with loop functions based on 
spacelike separated, spatial loops, 
can be interpreted as topological charges with values
in $\ZZ$. The fixed central element, multiplying the 
linking numbers, sets their scale, which in general depends on
the underlying theory.

\section{Causal commutators and the mass spectrum}
\label{sec4}

We turn now to the question under which 
circumstances the causal commutators can be different from zero. 
An answer is provided by relating 
it to the mass spectrum of the underlying theory. To this end 
we assume that the polynomial algebra $\fP$ of local field 
operators is irreducibly represented on the vacuum Hilbert 
space (sector) $\Hc$, cf.~the Wightman framework \cite{StWi}. 
On $\Hc$ there acts a continuous, unitary 
representation $U$ of the spacetime translations $\RR^4$, 
implementing the action of the translations on the local field operators,
$$
U(x) F(f) U(x)^{-1} = F(f_x) \, , \quad x \in \RR^4 \, ,
$$
and analogously for the local field $G$. The
representation $U$ satisfies the relativistic spectrum condition,  
\ie the joint spectrum of its generators $P$ has support in 
the forward lightcone $V_+$, and there is a unique, 
one-dimensional, translational 
invariant subspace of $\Hc$; it is generated by the vacuum vector 
$\Omega \in \Hc$ and lies in the domain of the elements of
$\fP$. 

\medskip 
These assumptions imply that the causal commutators,
being central elements of $\fP$ according to Proposition 
\ref{prop.2.2}, 
are represented on $\Hc$ by multiples of the identity, which 
coincide with their vacuum expectation values, 
$$
[A(h), B(k)] = \langle \Omega, [A(h), B(k)] \Omega \rangle \, 1 \, , 
\quad \supp \, h \perp \supp \, k \, .
$$
It enables us to determine their dependence on the
mass spectrum by an application
of the Jost-Lehmann-Dyson representation 
to the underlying local field operators $F,G$, cf.\ for example
\cite[Lem.~6.2]{DoHaRo}. Namely, 
for any pair of test functions $f,g \in \Dc_2(\RR^4)$, 
having support in spacelike separated double cones, one 
has\footnote{In \cite{DoHaRo} this relation was established in a 
framework of bounded local operators. Since the vector
$\Omega$ lies in the domain of the field operators, the argument
given there extends to the present case.}
$$
\langle \Omega, F(f) E(\Delta) G(g) \Omega \rangle 
= \langle \Omega, G(g) E(\Delta) F(f) \Omega \rangle, 
$$
where $E(\Delta)$ is the spectral projection of 
the mass operator $M = (P^2)^{1/2}$ for any given  
Borel set $\Delta \subset \RR$. It follows by 
integration that this relation still holds if 
one replaces $E(\Delta)$ by $c(M) E(\Delta)$, where $c(M)$ is any
continuous, bounded function of the mass operator. Thus, 
the vacuum expectation values of the commutators of local 
fields still vanish at spacelike distances 
if one arbitrarily rescales the 
intermediate states for different values of the mass. 

\medskip 
With this information we can turn now to the analysis of the
causal commutators. Given any 
test function $h \in \Cc_1(\RR^4)$, one has $\square h \in \Cc_1(\RR^4)$,
where $\square$ is the d'Alembertian. Since 
$\partial_\nu h^\nu = 0$, it follows that 
$\partial_\nu(\partial^\nu h^\mu - \partial^\mu h^\nu) = \square h^\mu$
and consequently $A(\square h) = A_\mu(\square h^\mu) =  
F_{\mu \nu}(\partial^\nu h^\mu - \partial^\mu h^\nu) = F(d h)$.
Similarly, one obtains for $k \in \Cc_1(\RR^4)$ the
equality $B(\square k) = 
G(d k)$. 
Now given any Borel set $\Delta_0 \subset \RR$ which has a finite
distance from $0 \in \RR$, it follows from the preceding 
result, the covariance of the fields,  and
the invariance of $\Omega$  under translations that 
$$
M^2 E(\Delta_0) A(h) \Omega = 
- E(\Delta_0) A(\square h) \Omega = - E(\Delta_0) F(d h) \Omega \, .
$$  
In view of the choice of $\Delta_0$, this implies 
$E(\Delta_0) A(h) \Omega = - M^{-2} E(\Delta_0) F(d h) \Omega$ 
and, in the same manner, one obtains 
$E(\Delta_0) B(k) \Omega = - M^{-2} E(\Delta_0) G(d k) \Omega$. 
In view of the hermiticity of $F,G$, and hence of $A,B$,  we 
therefore arrive at 
\begin{align*}
& \langle \Omega, A(h) E(\Delta_0) B(k) \, \Omega \rangle 
 - \langle \Omega, B(k) E(\Delta_0) A(h) \, \Omega \rangle  \\
= & \ \langle \Omega, F(dh) M^{-4} E(\Delta_0) \, G(dk) \, \Omega \rangle 
 - \langle \Omega, G(dk) M^{-4} E(\Delta_0) \, F(dh) \, \Omega \rangle \, .
\end{align*}
Now if $h,k \in \Cc_1(\RR^4)$ have spacelike separated supports, 
the same is true for their curls 
$d h, d k \in \Dc_2(\RR^4)$. Since $F,G$ are local fields, it 
follows from the preceding equality and the above 
relation, based on the 
Jost-Lehmann-Dyson representation, that 
$$
\langle \Omega, A(h) E(\Delta_0) B(k) \, \Omega \rangle 
 - \langle \Omega, B(k) E(\Delta_0) A(h) \, \Omega \rangle = 0 
\ \ \mbox{if} \ \ \supp \, h \perp \supp \, k \, . 
$$
In view of the continuity 
properties of the spectral resolution \mbox{$\Delta_0 \mapsto E(\Delta_0)$}  
and the fact that the contributions of the intermediate 
vacuum state $\Omega$ cancel in the commutator function, 
we have established the following result.
\begin{proposition}
In any theory having a mass spectrum
with values $m > 0$ on the orthogonal 
complement of the vacuum $\Omega$,
one has $[A(h), B(k)] = 0$ for all test functions
$h,k \in \Cc_1(\RR^4)$ with \mbox{$\supp \, h \perp \supp \, k$}. The 
commutators can be different from $0$ only in theories,  
where the spectral projection~$E(\{ 0 \})$ of the mass operator 
acts non-trivially on the orthogonal complement of $\Omega$. 
\end{proposition}

So we conclude that non-zero causal commutators, exhibiting the
linking number of the underlying loop functions, only appear
in the presence of massless particles; states with non-zero mass do not 
contribute to them.  

\section{Examples of causal commutators }
\label{sec5}

Having seen that the existence of non-trivial causal commutators 
is related to the existence of massless particles, 
we can easily exhibit now such examples. As was shown in the preceding
sections, it is sufficient to consider the vacuum expectation 
values of the commutators.

\medskip 
It follows from locality and Poincar\'e covariance 
of the fields  $F,G$ and 
the spectral properties of the translation 
operators that the (distributional)  
commutator functions on $\RR^4 \times \RR^4$ have a 
K\"all\'en-Lehman representation of the form 
$$
\langle \Omega, [F_{\mu \nu}(x), G_{\rho \sigma}(y) ] \Omega \rangle = \sum_Q
\!  \!
\int \! d \sigma_Q(m) \!  \! 
\int \! dp \, \varepsilon(p_0) \delta(p^2 - m^2) 
Q_{\mu \nu \rho \sigma}(p) e^{-ip(x-y)} \, ,
$$
where the sum extends over tensors $Q$ of rank four,   
built from polynomials in $p \in \RR^4$ and the metric 
tensor $g$, and $d \sigma_Q(m)$ are (unbounded)  
measures which have support on the mass
spectrum of the theory. 

\medskip 
The assumption that $F,G$ are hermitean, covariant and closed skew symmetric  
tensor fields of rank two imposes constraints on the tensors $Q$. 
We do not need to discuss here these constraints in full
generality since we know from the outset that contributions to the 
commutator functions with mass values $m>0$ 
do not affect the linking numbers. So we can focus on those
contributions to the above representation, which arise from an 
atomic value $m=0$ in the mass spectrum. They can be combined into
the expression
\begin{equation}
\label{e.5.1}
K^{(0)}_{\mu \nu \rho \sigma}(x-y) \doteq 
\int \! dp \, \varepsilon(p_0) \delta(p^2) 
\, Q^{(0)}_{\mu \nu \rho \sigma}(p) \, e^{-ip(x-y)} \, , 
\end{equation} 
where the tensor $ Q^{(0)}$ is also built from  
polynomials in $p \in \RR^4$ and the metric 
tensor $g$. It follows after a straight forward
computation that the above mentioned properties of the fields $F,G$ 
imply that $Q^{(0)}$ must have the form
\begin{align}
\label{e.5.2}
Q^{(0)}_{\mu \nu \rho \sigma}(p) & =
c_1 \, (p_\mu p_\rho g_{\nu \sigma}  - p_\nu p_\rho g_{\mu \sigma}
- p_\mu p_\sigma g_{\nu \rho} + p_\nu p_\sigma g_{\mu \rho}) \notag\\
& + c_2 \, (p_\mu p_\tau g_{\nu \upsilon}  - p_\nu p_\tau g_{\mu \upsilon}
- p_\mu p_\upsilon g_{\nu \tau} + p_\nu p_\upsilon g_{\mu \tau}) \ 
\varepsilon^{\tau \upsilon}_{\, \rho \sigma} \, , 
\end{align}
where $c_1, c_2 \in \RR$ and $\varepsilon$ is the 
totally skew symmetric Levi-Civita tensor. 

\medskip 
The tensor in 
the first line is familiar from expectation values of the 
electromagnetic field and complies with all
constraints, irrespective of the underlying mass spectrum.
Its contribution to $K^{(0)}$ in \eqref{e.5.1} is 
proportional  to the commutator of the free electromagnetic field
$F^{(0)}$ with itself. The tensor in the second line
requires some comment, however. It has the correct hermiticity properties,
is skew symmetric in $\mu, \nu$ as well as $\rho, \sigma$ and manifestly
encodes the fact that $F$ is closed. That this condition 
is also satisfied 
for $G$ follows from the fact that the tensor~$Q^{(0)}$ 
is restricted in $K^{(0)}$ to the mass shell $p^2 = 0$. 
The contribution to $K^{(0)}$, resulting from the 
second line in relation \eqref{e.5.2},  is therefore 
proportional to the commutator of the free 
electromagnetic field~$F^{(0)}$ with its Hodge-dual 
$\star F^{(0)}$. But this dual tensor field is also closed
since the electric current vanishes, \ie 
$d \star F^{(0)} = \star \delta  F^{(0)} = 0$.  
Thus any tensor $Q^{(0)}$ of the form given above 
gives rise to an admissible contribution $K^{(0)}$ 
to commutator functions.

\medskip
As has been shown by Roberts \cite{Roberts}, cf.\ also \cite{BuCiRuVa2}, 
the commutator of the free electromagnetic 
field $F^{(0)}$ and its Hodge dual $\star F^{(0)}$ 
gives rise to non-trivial causal commutators for 
the corresponding intrinsic vector potentials; in contrast,  
the causal commutators, determined by the 
commutator of the free electromagnetic
field with itself, vanish \cite{BuCiRuVa2}. In view of the 
preceding results, we have thus arrived at the following 
proposition, characterizing all fields $F,G$
leading to intrinsic vector potentials with non-trivial 
causal commutators. 

\begin{proposition}
\label{prop.5.1}
Let $F,G$ be local, hermitean, covariant and 
closed skew symmetric tensor fields. The causal commutators of the
corresponding intrinsic vector potentials $A,B$ are different from zero, 
indicating the linking numbers of the underlying loop functions, 
if and only if the (distributional) commutator function 
$ \langle \Omega, [F( \cdot ), G( \cdot ) ] \, \Omega \rangle $
contains in its K\"all\'en-Lehmann representation a contribution
$K^{(0)}(\cdot)$ as in \eqref{e.5.1}, where $Q^{(0)}$ is of 
the form~\eqref{e.5.2} with $c_2 \neq 0$. 
\end{proposition}

\medskip
We conclude this section by noting that there exists an 
abundance of quantum field theory models of fields $F,G$ with 
properties described in the preceding proposition. Examples 
can be easily exhibited in the class of generalized free field
theories. There the fields have c-number commutation
relations, so a generalized free field theory is completely fixed by 
specifying the two-point functions of the 
fields in the vacuum state. Thus one may put, for example, 
$$
\langle \Omega, F(f) G(g) \Omega \rangle 
\doteq c \, \langle \Omega, F^{(0)}(f) \star \! F^{(0)}(g) \Omega \rangle 
\, ,  \quad f,g \in \Dc_2(\RR^4) \, ,
$$
where $c \in \RR \backslash \{ 0 \}$. One then chooses with the help of
the K\"all\'en-Lehmann representation arbitrary two point functions for the 
field $F$, respectively $G$, satisfying
\begin{align*}
\langle \Omega, F(\overline f) F(f) \Omega \rangle 
& \geq |c| \, \langle \Omega, F^{(0)}(\overline f) F^{(0)}(f) \Omega \rangle 
\geq 0 \, , \quad  f \in \Dc_2(\RR^4) \, , \\
\langle \Omega, G(\overline g) G(g) \Omega \rangle 
& \geq |c| \, \langle \Omega, F^{(0)}(\overline g) \, F^{(0)}(g) \Omega \rangle 
\geq 0 \, , \, \quad g \in \Dc_2(\RR^4) \, .
\end{align*}
Since 
$\langle \Omega, \star F^{(0)}(\overline g) \, \star \! F^{(0)}(g) \Omega \rangle
= \langle \Omega, F^{(0)}(\overline g) \, F^{(0)}(g) \Omega \rangle$ 
for $g \in \Dc_2(\RR^4)$, it implies  
$$
|\langle \Omega, F(\overline f) G(g) \Omega \rangle|^2 
\leq \langle \Omega, F(\overline f) F(f) \Omega \rangle 
\langle \Omega, G(\overline g) G(g) \Omega \rangle \, ,
\quad f,g \in \Dc(\RR^4) \, , 
$$
in accordance with the condition of Wightman positivity \cite{StWi}. 
The resulting generalized free field theories comply with 
all constraints on the tensor fields $F,G$, 
and the causal commutators of the corresponding intrinsic 
vector potentials are different from zero.

\section{Conclusions}
\label{sec6}

In the present investigation we have studied the appearance of
linking numbers in quantum field theory, which arise from closed 
skew symmetric tensor fields $F,G$. These linking numbers appear 
as ``superselected charges'' in commutators
of the corresponding intrinsic vector potentials $A,B$, which are
smeared with loop functions, having supports in spacelike separated,
spatial loops. They are different from zero only in the presence
of massless particles in the theory. The intrinsic vector potentials 
$A,B$, which are defined on the space of vector valued test 
functions~$\Cc_1(\RR^4)$ with vanishing divergence, 
can then not both be extended to 
local pointlike vector fields, defined on the space of all vector valued
test functions $\Dc_1(\RR^4)$; this is true even if one admits potentials in 
indefinite metric spaces. So such fields provide genuine examples of 
closed, but \textit{not} exact tensor fields in the framework of local 
quantum field theory. 

\medskip
One can overcome this cohomological obstruction by relaxing the 
condition of pointlike locality for the vector potentials, admitting 
fields which are ray-localized. Such an approach   
has been proposed by J.~Mund, B.~Schroer and others, cf.~\cite{MuOl} 
and references quoted there. There 
one proceeds from closed skew symmetric tensor fields $F$
to vector potentials, defined as operator-valued 
distributions on $\RR^4 \times dS^3$,
where $dS^3 = \{ e \in \RR^4 : e^2 = -1 \}$ is de Sitter space.
These potentials are given by  
$$
A_{e \, \mu}(x) \doteq \int_0^\infty \!  du \, e^\rho F_{\mu \rho}(x + ue) \, .
$$ 
They satisfy, as desired,  
\begin{align*}
\partial_\mu A_{e \, \nu}(x) - \partial_\nu A_{e \, \mu}(x)
& = \int_0^\infty \!  du \, e^\rho (\partial_\mu F_{\nu \rho}(x + ue) 
- \partial_\nu F_{\mu \rho}(x + ue))  \\
& = - 
\int_0^\infty \!  du \, e^\rho \partial_\rho 
F_{\mu \nu}(x + ue) = F_{\mu \nu}(x) \, ,
\end{align*}
where it is assumed that $F$ vanishes at spacelike infinity
in the states of interest. The locality properties of the underlying
tensor fields then imply that the potentials 
$A_e(x)$ commute with all local fields which 
are localized in the spacelike complement of the 
ray $x + \RR_+ e$. This approach yields gauge invariant 
vector potentials, which are defined on the test function space 
$\Dc_1(\RR^4) \times \Dc(dS^3)$. The price one has 
to pay is to give up the standard locality property of the
vector potentials, which is, however, unavoidable in the cases 
considered in the present article. It is an intriguing question 
whether these ray localized vector potentials can serve as a
substitute for the local vector potentials in gauge quantum field 
theory, as envisaged in \cite{Sch}. 

\medskip
Let us conclude with the remark that our  
analysis of linking numbers in commutator functions is based on 
purely topological arguments, avoiding computations
of the Gauss integral~\eqref{eq.gauss}. 
We therefore hope that by refining our arguments 
one can establish results on knot invariants and Jones polynomials,  
as obtained by Witten in topological quantum field theory \cite{Witten},  
also in case of relativistic quantum field theories with non-abelian
gauge groups. Since the notion of homology, used in Lemma~\ref{lem.3.2}, 
is insensitive to knots, such an analysis has to be based on the 
finer concept of \emph{isotopy} \cite{BaMu,Spera}, however.

\section*{Acknowledgement}

\vspace*{-2mm}
DB  gratefully acknowledges the hospitality and support 
extended to him by Roberto Longo and 
the University of Rome ``Tor Vergata'', which made
this collaboration possible. 
FC and GR are supported by the ERC Advanced Grant 
669240 QUEST ``Quantum Algebraic Structures and Models". 
EV is supported in part by OPAL ``Consolidate the Foundations''.


\begin{thebibliography}{22}
{\small 

\bibitem{BaMu} Baez, J., Muniain, J.P., 
{\it Gauge Fields, Knots And Gravity},
Series on knots and everything Vol.4,
World Scientific, Singapore, New Jersey, London, Hong Kong, 1994

\bibitem{BuCiRuVa1} Buchholz, D., Ciolli, F., Ruzzi, G. and Vasselli, E., 
``The universal C*-algebra of the electromagnetic field'',
Lett.\ Math.\ Phys.\ {\bf 106} (2016) 269--285 \ 
Erratum: Lett.\ Math.\ Phys.\ {\bf 106} (2016) 287

\bibitem{BuCiRuVa2} Buchholz, D., Ciolli, F., Ruzzi, G. and Vasselli, E., 
``The universal C*-algebra of the electromagnetic field II.
Topological charges and spacelike linear fields'',
Lett.\ Math.\ Phys.\ {\bf 107} (2017) 201--222 \ 

\bibitem{DoHaRo} Doplicher,~S.\ Haag,~R.\ and Roberts,~J.E., 
``Local observables and particle statistics II'',
Commun.\ Math.\ Phys.\ {\bf 35} (1974) 49-85 

\bibitem{Hatcher} Hatcher, A.,  
{\it Algebraic Topology}, 
Cambridge University Press, Cambridge England, 2002 

\bibitem{MuOl} Mund,~J. and de Oliveira,~E.T., 
``String–localized free vector and tensor potentials for
massive particles with any spin: I. Bosons'', 
Comm.\ Math.\ Phys.\ {\bf 355} (2017) 1243–1282 

\bibitem{Roberts} Roberts, J.E., 
``A survey of local cohomology'', \ 
In: Mathematical problems in theoretical physics. (Rome, 1977), 81--93, 
Lecture Notes in Phys. 80, Springer, Berlin, New York, 1978

\bibitem{Rolfsen} 
Rolfsen, D. {\it Knots and Links},  
AMS Chelsea Publishing, Providence, RI, 1990 

\bibitem{Sch}
Schroer, B., ``Beyond gauge theory: positivity and causal localization 
in the presence of vector mesons'', 
Eur.\ Phys.\ J.\ C \textbf{76} (2016) 378

\bibitem{Spera} Spera, M.,
``A Survey on the Differential and Symplectic Geometry of Linking Numbers'', \ 
Milan J.\ Math.\ {\bf 74} (2006) 139--197

\bibitem{StWi}  Streater, R.F.\ and Wightman, A.S., 
{\it PCT, Spin and Statistics, and All That},
W.A. Benjamin, New York, Amsterdam, 1964

\bibitem{Sw} Swieca, J.A.,
``Charge  screening  and  mass  spectrum'',  
Phys.\  Rev.\ {\bf D13} (1976) 312-314

\bibitem{Witten} Witten, E.,
``Quantum field theory and the Jones polynomial'',
Comm.\ Math.\ Phys.\ {\bf 121} (1989) 351--399

}
\end{thebibliography}
\end{document}